\newif\ifreport

\reportfalse

\ifreport
\documentclass[journal]{IEEEtran}
\else
\documentclass[conference]{IEEEtran}
\renewcommand\baselinestretch{1}
\fi

\usepackage{multirow}
\usepackage{booktabs}
\usepackage{threeparttable}
\usepackage{hyperref}
\usepackage{cite}
\usepackage{graphicx}
\usepackage{amsmath}
\usepackage{amssymb}
\usepackage{amsthm}
\usepackage{color}
\usepackage[noblocks]{authblk}
\theoremstyle{definition}

\newtheorem{theorem}{Theorem}

\newtheorem{corollary}{Corollary}

\def\blue{\color{black}}

\IEEEoverridecommandlockouts

\begin{document}

\title{Analog Beam Tracking in Linear Antenna Arrays: Convergence, Optimality, and Performance}

\author{Jiahui Li$^*$, Yin Sun$^\S$, Limin Xiao$^{\P\ddagger}$, Shidong Zhou$^*$, C. Emre Koksal$^\dagger$ \\
$^*$Dept. of EE, $^\P$Research Institute of Information Technology, Tsinghua University, Beijing, 100084, China\\
$^\S$Dept. of ECE, Auburn University, Auburn AL, 36849, U.S.A.\\
$^\dagger$Dept. of ECE, The Ohio State University, Columbus OH, 43210, U.S.A.
\thanks{$^\ddagger$Corresponding author.}
\thanks{J. Li, L. Xiao, and S. Zhou have been supported in part by National Basic Research Program of China (973 Program) grant 2012CB316002, National S\&T Major Project grant 2017ZX03001011-002, National Natural Science Foundation of China grant 61631013, National High Technology Research and Development Program of China (863 Program) grant 2014AA01A703, Science Fund for Creative Research Groups of NSFC grant 61321061, Tsinghua University Initiative Scientific Research grant 2016ZH02-3, International Science and Technology Cooperation Program grant 2014DFT10320, "2011 Plan" Wireless Communication Technology Co-Innovation Center grant 20161210020, Tsinghua-Qualcomm Joint Research Program, and Huawei HIRP project. Y. Sun has been supported in part by ONR grant N00014-17-1-2417.}}


\maketitle

\begin{abstract}


The directionality of millimeter-wave (mmWave) communications creates a significant challenge in serving fast-moving mobile terminals on, e.g., high-speed vehicles, trains, and UAVs. This challenge is exacerbated in mmWave systems using analog antenna arrays, because of the inherent non-convexity in the control of the phase shifters. In this paper, we develop a recursive beam tracking algorithm which can simultaneously achieve fast tracking speed, high tracking accuracy, low complexity, and low pilot overhead. In static scenarios, this algorithm converges to the minimum Cram\'er-Rao lower bound (CRLB) of beam tracking with high probability. In dynamic scenarios, even at SNRs as low as 0dB, our algorithm is capable of tracking a mobile moving randomly at an absolute angular velocity of 10-20 degrees per second, using only 5 pilot symbols per second. If combining with a simple TDMA pilot pattern, this algorithm can track hundreds of high-speed mobiles in 5G configurations. Our simulations show that the tracking performance of this algorithm is much better than several state-of-the-art algorithms.

 
\end{abstract}

\bstctlcite{BSTcontrol}

\section{Introduction}\label{sec_intro}

{
The explosively growing data traffic in future wireless systems can be leveraged by using large antenna arrays and millimeter-wave (mmWave) frequency band \ifreport\else\cite{Larsson2014massive, Heath2016overview}\fi. However, as the array size grows and the carrier frequency increases, the large number of A/D (or D/A) converters in the fully digital array make the design infeasible due to high energy consumption and huge hardware cost \cite{Heath2016overview}. A promising alternative is analog beamforming \ifreport\else\cite{Han2015Large, Puglielli2016Design, Heath2016overview}\fi, in which the signals of all antennas are beamformed in the analog domain by using phase shifters, and a single A/D (or D/A) is used for digital processing. This analog beamforming solution has been standardized  by IEEE 802.11ad \cite{IEEE80211ad} and IEEE 802.15.3c \cite{IEEE802153c}, and is actively discussed by several 5G industrial organizations \cite{METIS2015, ITU2015}.
}

{

One fundamental challenge in analog beamforming is how to track the beam directions using limited pilot resources. This challenge is especially difficult when a large number of narrow beams generated from many fast-moving mobile terminals or reflectors need to be tracked. This challenge has been recognized in the industry as one important research task for 5G massive MIMO and mmWave systems, e.g., \cite{Brown2016Promise}.

There has been a number of recent studies on beam direction estimation/tracking for analog beamforming \ifreport\else\cite{Hur2013Millimeter, Alkhateeb2014Channel, Alkhateeb2015Compressed, zhang2016mobile, palacios2016tracking, garcia2017optimal, bae2017new}\fi. In \ifreport\else\cite{Hur2013Millimeter, Alkhateeb2014Channel, Alkhateeb2015Compressed, zhang2016mobile, palacios2016tracking}\fi, one round of beam sweeping, scanning many spatial beam directions in a codebook, is needed for updating the beam direction estimate. In \cite{garcia2017optimal, bae2017new}, the training is performed based on knowledge of prior beam estimates, which is beneficial to save pilots. However, the optimal training has not been obtained and the estimation/tracking is not optimized accordingly, which leads to poor tracking accuracy. 


The goal of this paper is to develop an efficient beam tracking algorithm that can track a large number of high-speed mobiles with high accuracy and low pilot overhead.}
The main contributions of this paper are summarized as follows:
\begin{itemize}
\item 
We develop a recursive beam tracking algorithm. In \emph{static} beam tracking scenarios, its convergence and asymptotic optimality are established in three steps:
First, we prove that it converges to a set of stable beam directions with probability one (Theorem \ref{th_convergence}). Second, we prove that under certain conditions, it converges to the real beam direction, instead of other sub-optimal stable directions, with high probability (Theorem \ref{th_lock}). Finally, if the step-sizes are chosen appropriately, then with high probability, the mean square error (MSE) of the proposed algorithm converges to the minimum\footnote{The CRLB is a function of the beamforming control action. The minimum CRLB is obtained by optimizing among all control actions (see Section \ref{sec_problem}).} CRLB (Theorem \ref{th_normal}).

\item Our simulation results in both \emph{static} and \emph{dynamic} beam tracking scenarios suggest that the proposed algorithm can achieve much lower beam tracking error and higher data rate than several state-of-the-art algorithms, with the same pilot overhead. Particularly, if 5 uniformly inserted pilot symbols per second are used and the receive SNR of each antenna is 10 dB (or 0 dB), by combining with a TDMA round-robin pilot pattern, the proposed algorithm can track 1000 narrow beams each rotating at an angular velocity of \textbf{18.33$^\circ$/s} (or \textbf{13.18$^\circ$/s}), which is \textbf{72 miles/h} (or \textbf{52 miles/h}) if the transmitters/reflectors steering these beams are at a distance of 100 meters.  And when it is needed to track extremely fast mobiles, one can insert more pilot symbols for each mobile. Hence, the tracking speed can be \textbf{very fast}.

\end{itemize}

Two major technical reasons why our algorithm achieves a good tracking performance are: (i) the probing beamforming direction in each time-slot is close to the real beam direction, while the other algorithms (e.g., \cite{Hur2013Millimeter, Alkhateeb2014Channel, Alkhateeb2015Compressed, zhang2016mobile, palacios2016tracking}) probe a lot of beam directions, and (ii) an optimal step-size is chosen to ensure a fast convergence rate to the global optimal beam direction, instead of other local optimal beam directions. To the extent of our knowledge, this paper presents the first theoretical analysis on the convergence and asymptotic optimality of analog beam tracking in antenna array systems.

The remaining of this paper is organized as follows: In Sections \ref{sec_model}, the system model is introduced. In Sections \ref{sec_problem} and \ref{sec_algorithm}, we formulate the beam tracking problem and develop a recursive beam tracking algorithm that is proved to be asymptotically optimal with high probability in \emph{static} beam tracking scenarios. In Section \ref{sec_simulation}, we evaluate its performance in both \emph{static} and \emph{dynamic} beam tracking scenarios. 


\ifreport
\else
\section{System Model}\label{sec_model}

\fi
\subsection{System Model}
Consider the linear antenna array receiver in Fig. \ref{fig_system}, where $M$ antennas are placed along a line, with a distance $d$ between neighboring antennas. The antennas are connected by phase shifters to a single radio frequency (RF) chain, and the phase shifters are controlled to steer the observation direction. 
In time-slot $n$,  a pilot signal arrives at the antenna array from an angle-of-arrival (AoA) $\theta_n\in[-\pi/2,\pi/2]$. Hence, the steering vector of this arriving beam is
\vspace{-1mm}
\begin{equation}\label{eq_steer}
	\textbf{a}(x_n) = \left[ 1, e^{j \frac{2\pi d}{\lambda} x_n}, \cdots, e^{j \frac{2\pi d}{\lambda}(M-1)x_n} \right]^\text{H},
\vspace{-1mm}
\end{equation}
where $x_n = \sin(\theta_n)$ is the sine of the AoA $\theta_n$ and $\lambda$ is the wavelength. The channel response is $\textbf{h}(x_n) = \beta \textbf{a}(x_n)$, where $\beta$ is the complex channel coefficient. 

\begin{figure}
\centering
\includegraphics[width=7cm]{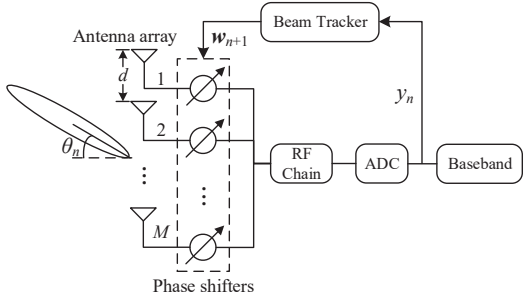}
\vspace{-2.75mm}
\caption{System model.}
\vspace{-4.5mm}
\label{fig_system}
\end{figure}

Let $w_{mn} \in[-\pi,\pi]$ be the phase shift in radians provided by the $m$-th phase shifter in time-slot $n$. Then, the analog beamforming vector steered by the phase shifters is
\vspace{-1mm}
\begin{equation}\label{eq_bf}
	\textbf{w}_n = \frac{1}{\sqrt{M}}\left[ e^{jw_{1n}}, e^{jw_{2n}}, \cdots, e^{jw_{Mn}} \right]^\text{H}.
\vspace{-1mm}
\end{equation}
Combining the output signals of the phase shifters and dividing the summed signal by $\beta$ yields 
\vspace{-1mm}
\begin{equation}\label{eq_observation}
	y_n = \textbf{w}_n^\text{H}\textbf{a}(x_n) + \frac{{z}_n}{\sqrt{\rho}},
\vspace{-1mm}
\end{equation}
where $\rho = {|\beta|^2}/{\sigma^2}$ is the SNR at each antenna, $\sigma^2$ is the noise power, and the ${z}_n$'s are \emph{i.i.d.} circularly symmetric complex Gaussian random variables with zero mean and unity variance. Given $x_n$ and $\textbf{w}_n$, the conditional probability density function of $y_n$ is
\vspace{-1mm}
\begin{equation}\label{eq_pdf}
	p(y_n| x_n, \textbf{w}_n) = \frac{\rho}{\pi} e^{- \rho\left| y_n -  \textbf{w}_n^\text{H}\textbf{a}(x_n) \right|^2}.
\vspace{-1mm}
\end{equation}

A beam tracker determines the analog beamforming vector $\textbf{w}_{n}$ and provides an estimate $\hat{x}_n$ of the sine $x_n$ of the AoA.\footnote{Interestingly, by tracking the sine $x_n$, we obtain a beam tracking algorithm with lower complexity and higher robustness than tracking the AoA $\theta_n$ \cite{Li2017analog}.} From a control system perspective, $x_n$ is the system state, $\hat{x}_n$ is the estimate of the system state, the beamforming vector $\textbf{w}_{n}$ is the control action, and $y_n$ is the observation. Let $\psi = (\textbf{w}_1, \textbf{w}_2, \ldots, \hat{x}_1, \hat{x}_2, \ldots)$ represent a beam tracking policy. In particular, we consider the set $\Psi$ of \emph{causal} beam tracking policies: At the end of time-slot $n$, the estimate $\hat{x}_n$ of time-slot $n$ and the control action $\textbf{w}_{n+1}$ of time-slot $n+1$ are determined  by using the history of the control actions $(\textbf{w}_1, \ldots, \textbf{w}_n)$ and the observations $(y_1, \ldots, y_n)$.



\section{Problem Formulation and Performance Bound}\label{sec_problem}
Given any time-slot $n$, the beam tracking problem can be formulated as 
\begin{align}\label{eq_problem}
	\underset{\psi \in \Psi}{\min}~& \mathbb{E}\left[ \left( \hat{x}_{n} - x_n \right)^2 \right] \\ 
	\text{s.t.}~&  \mathbb{E}\left[ \hat{x}_{n} \right] = x_n, \label{eq_constrant}\\
	& \eqref{eq_steer},~\eqref{eq_bf},~\eqref{eq_observation}, \nonumber
\end{align}
where the constraint \eqref{eq_constrant} ensures that $\hat{x}_{n}$ is an \emph{un-biased} estimate of $x_n$. Problem \eqref{eq_problem} is a constrained sequential control and estimation problem that is difficult, if not impossible, to solve optimally. First, the system is partially observed through the observation $y_n$. Second, both the control action $\textbf{w}_{n}$ and the estimate $\hat{x}_n$ need to be optimized in Problem \eqref{eq_problem}: On the one hand, because only the phase shifts $(w_{1n}, \ldots, w_{Mn})$ in \eqref{eq_bf} are controllable, the optimization of $\textbf{w}_{n}$ is a non-convex optimization problem. On the other hand, as shown in Fig. \ref{fig_imagpart} and  \eqref{eq_stable_points} below, the optimization of the estimate $\hat{x}_n$ is also non-convex and there are multiple local optimal estimates.




Next, we consider \emph{static} beam tracking scenarios, where $x_n=x$ for all time-slot $n$, and establish a lower bound of the MSE in \eqref{eq_problem}: Given the control actions $(\textbf{w}_1,\ldots,\textbf{w}_n)$, the MSE is lower bounded by the CRLB \cite{nevel1973stochastic}
\vspace{-1mm}
\begin{equation}\label{eq_MMSE}
	\begin{aligned} & \mathbb{E}\left[ \left( \hat{x}_{n} - x \right)^2 \right] \ge \frac{1}{\sum_{i=1}^n I(x, \textbf{w}_i)},
	\end{aligned}
\vspace{-1mm}
\end{equation}
where $I(x, \textbf{w}_i) $ is the Fisher information \cite{Poor1994estimation} that can be computed by using \eqref{eq_pdf}:
\vspace{-1mm}
\begin{equation}\begin{aligned} \label{eq_fisher_sub}
	 I(x, \textbf{w}_i) &= \mathbb{E}\left[ \left. - \frac{\partial^2 \log p \left( y_i|x, \textbf{w}_i \right)}{\partial x^2} \right| x, \textbf{w}_i \right]  \\
	& = \frac{2\rho}{M}\left| \sum\limits_{m=1}^{M} \frac{2\pi d}{\lambda}(m-1) e^{j\left[w_{mi} - \frac{2\pi d}{\lambda}(m-1)x\right]} \right|^2. 
\end{aligned}\vspace{-1mm}
\end{equation}
By optimizing the control actions $(\textbf{w}_1,\ldots,\textbf{w}_n)$ in the right-hand-side (RHS) of \eqref{eq_MMSE}, we obtain
\vspace{-1mm}
\begin{equation}\label{eq_opt_MMSE}
\frac{1}{n} \sum_{i=1}^n I(x, \textbf{w}_i) \le \frac{2M(M-1)^2\pi^2 d^2 \rho}{\lambda^2} \overset{\Delta}{=} I_{\max},
\vspace{-1mm}
\end{equation}
where the maximum Fisher information $I_{\max}$ in (\ref{eq_opt_MMSE}) is achieved if, and only if, for $i = 1, \ldots, n$
\vspace{-1mm}
\begin{equation}\label{eq_ctrl}
	\begin{aligned}
	\textbf{w}_i = \frac{\textbf{a}(x)}{\sqrt{M}} = \frac{1}{\sqrt{M}}\left[ 1, e^{j \frac{2\pi d}{\lambda} x}, \cdots, e^{j \frac{2\pi d}{\lambda}(M-1)x} \right]^\text{H}.
	\end{aligned}
\vspace{-1mm}
\end{equation}
Hence, the MSE is lower bounded by the minimum CRLB
\vspace{-1mm}
\begin{align}\label{eq_CRLB}
	\mathbb{E}\left[ \left( \hat{x}_{n} - x \right)^2 \right] \ge \frac{1}{nI_{\max}}.
\end{align}

\section{Recursive Analog Beam Tracking: Algorithm and Analysis}\label{sec_algorithm}

In this section, we design a recursive analog beam tracking algorithm and prove that its MSE converges to the lower bound on the RHS of  (\ref{eq_CRLB}) with high probability in \emph{static} beam tracking scenarios.

\subsection{Algorithm Design}\label{sec_alg_design}

We develop a recursive analog beam tracking algorithm that consists of two stages: 1) coarse beam sweeping and 2) recursive beam tracking.

\begin{figure}
\centering
\includegraphics[width=6.5cm]{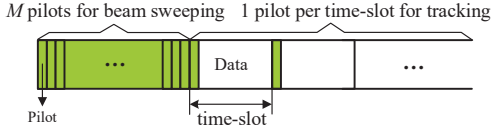}
\vspace{-2mm}
\caption{Frame structure.}
\vspace{-5mm}
\label{fig_frame}
\end{figure}

 \textbf{Recursive Analog Beam Tracking (Algorithm 1)}:


\begin{itemize}
\item[1)] \emph{Coarse Beam Sweeping:} 
Receive $M$ pilots successively (see Fig. \ref{fig_frame}). The analog beamforming vector $\tilde{\textbf{w}}_{m}$ for receiving the $m$-th pilot signal $\tilde{y}_m$ is 
\begin{align}\label{eq_codebook}
\tilde{\textbf{w}}_{m} = \frac{1}{\sqrt{M}}\textbf{a}\left(\frac{2m}{M} - \frac{M+1}{M}\right), m = 1, \ldots, M.
\end{align}
Find the initial estimate  $\hat{x}_0$ of the beam direction by
\begin{equation}\label{eq_initial}
\begin{aligned}
	\hat{x}_0 =\underset{\hat{x} \in \mathcal{X}}{\arg\max}~\left|\textbf{a}(\hat{x})^\text{H}\cdot \sum_{m=1}^M \tilde{y}_m \tilde{\textbf{w}}_{m} \right|,
\end{aligned}
\end{equation}
where $\mathcal{X} = \left\{\frac{1 - M_0}{M_0}, \frac{3 - M_0}{M_0}, \ldots, \frac{M_0-1}{M_0}\right\}$ and $M_0 \ge M$ determines the estimation resolution. 

\item[2)] \emph{Recursive Beam Tracking:} In each time-slot $n = 1, 2, \ldots$, one pilot is received at the beginning (see Fig. \ref{fig_frame}) using $\textbf{w}_{n}$, given by 
\begin{equation}\label{eq_est_ctrl}
	\begin{aligned}
	\textbf{w}_{n} = \frac{1}{\sqrt{M}}\textbf{a}\left(\hat{x}_{n-1}\right).
	\end{aligned}
\end{equation}
The estimate $\hat{x}_{n}$ of the beam direction is updated by 
\begin{equation}\label{eq_est}
\begin{aligned}
\hat{x}_{n} = \left[ \hat{x}_{n-1} - a_n \operatorname{Im}\left\{ y_{n} \right\} \right]_{-1}^1,
\end{aligned}
\end{equation}
where $[x]_{b}^{c} = \max\left\{ \min\{ x, c \}, b \right\}$ and $a_n > 0$ is the step-size that will be specified later. 
\end{itemize}

In \emph{Stage 1}, the exhaustive beam sweeping is used, and an initial estimate $\hat{x}_0$ is obtained in \eqref{eq_initial} by using the orthogonal matching pursuit method \cite{Alkhateeb2015Compressed}. Its resolution is adapted by the size $M_0$ of the dictionary $\mathcal{X}$, and a larger $M_0$ provides a more accurate estimate. Our simulations suggest that, if the SNR $\rho \geq 0$dB, $M_0 = 4M$, and $M = 16$, a good initial estimate $\hat{x}_0$ within the mainlobe $\mathcal{B}(x_0)$ (e.g., see Fig. \ref{fig_imagpart}), defined by
\vspace{-0.5mm}
\begin{equation}\label{eq_mainlobe}
	\mathcal{B}(x_0) =\left( x_0 - \frac{\lambda}{Md}, x_0 +  \frac{\lambda}{Md}\right) \bigcap [-1, 1],\vspace{-0.5mm}
\end{equation}
can be obtained with a probability higher than $99.99\%$.\footnote{One can use more time-slots (pilot resources) to support lower SNR in \emph{Stage 1}. As \emph{Stage 1} is executed only once, this will not increase the total pilot overhead by much.}

In \emph{Stage 2}, the estimate $\hat{x}_n$ and the control $\mathbf{w}_n$ are updated recursively to realize an accurate tracking performance. The recursive beam tracker in \eqref{eq_est} is motivated by the following maximum likelihood (ML) problem:
\vspace{-0.5mm}
\begin{equation}
\label{eq_ML_estimator}\underset{\hat{x}_n}{\max}\left\{\underset{\textbf{w}_n}{\max}~\sum_{i=1}^n\mathbb{E}\bigg[\log p\left( y_i|\hat{x}_n,\!\textbf{w}_i\!\right)\bigg| \begin{matrix} \hat{x}_n,\!\textbf{w}_1,\!\ldots,\!\textbf{w}_i, \\ \!y_1,\!\ldots,\!y_{i-1}\end{matrix}\!\bigg]\right\}.
\vspace{-0.5mm}
\end{equation}
Rather than directly solve \eqref{eq_ML_estimator}, we propose to use the stochastic Newton's method, given by \cite[Section 10.2]{nevel1973stochastic}
\vspace{-0.5mm}
\begin{align}\label{eq_sto_NT}\hat{x}_{n} &~= \left[ \hat{x}_{n-1} - a_{n} \cdot \frac{\frac{\partial \log p\left( y_n|\hat{x}_{n-1}, \textbf{w}_n \right)}{\partial \hat{x}_{n-1}}}{\mathbb{E}\left[\left.\frac{\partial^2 \log p\left( y_n|\hat{x}_{n-1}, \textbf{w}_n \right)}{\partial \hat{x}_{n-1}^2}\right| \hat{x}_{n-1}, \textbf{w}_n \right]} \right]_{-1}^1 \nonumber \\
&~= \left[ \hat{x}_{n-1} + a_{n} \cdot \frac{\frac{\partial \log p\left( y_n|\hat{x}_{n-1}, \textbf{w}_n \right)}{\partial \hat{x}_{n-1}}}{I(\hat{x}_{n-1}, \textbf{w}_n)} \right]_{-1}^1,\vspace{-0.5mm}\end{align}
where the control $\textbf{w}_n$ can be obtained by maximizing the Fisher information $I(\hat{x}_{n-1}, \textbf{w}_n)$, which yields \eqref{eq_est_ctrl}. By plugging \eqref{eq_pdf}, \eqref{eq_fisher_sub} and \eqref{eq_est_ctrl} into \eqref{eq_sto_NT}, we can derive the low complexity recursive beam tracker in \eqref{eq_est}.

{\blue \subsection{Multiple Stable Points for Recursive Procedure}

To obtain the points that the recursive procedure \eqref{eq_est_ctrl} and (\ref{eq_est}) might converge to, we will introduce its corresponding ordinary differential equation (ODE). Using \eqref{eq_observation} and \eqref{eq_est_ctrl}, the recursive beam tracker in (\ref{eq_est}) can also be expressed as
\begin{align}
\hat{x}_{n} = \left[ \hat{x}_{n-1} + a_n \left( f(\hat{x}_{n-1}, x) - \frac{\operatorname{Im}\left\{{z}_n\right\}}{\sqrt{\rho}} \right) \right]_{-1}^1,
\end{align}
where function $f: \mathbb{R} \times  \mathbb{R}\mapsto \mathbb{R}$ is defined as
\begin{equation}\label{eq_fx}
	\begin{aligned}
	f(v, x) \overset{\Delta}{=} - \frac{1}{\sqrt{M}}\operatorname{Im}\left\{\textbf{a}(v)^\text{H}\textbf{a}(x)\right\}.
	\end{aligned}
\end{equation}
This recursive procedure can be seen as a noisy, discrete-time approximation of the following ODE \cite[Section 2.1]{borkar2008stochastic}
\begin{align}\label{eq_ODE}
\frac{d \hat{x}(t)}{dt} \!=\!\left\{ \begin{array}{cl} \max\{f(-1, x),0\} &~\text{if}~\hat{x}(t) = -1 \\
f(\hat{x}(t),x) &~\text{if}~-1<\hat{x}(t) < 1 \\
\min\{f(1, x),0\} &~\text{if}~\hat{x}(t) = 1,\end{array}\right.
\end{align}
with $t \geq 0$ and $\hat{x}(0)\!=\!\hat{x}_0$. According to \cite{kushner2003stochastic, borkar2008stochastic}, the recursive procedure will converge to one of the stable points of the ODE \eqref{eq_ODE}. Here the stable point of the ODE \eqref{eq_ODE} is defined as a point $v_0$ that satisfies $f(v_0, x)\!=\!0$ and $f_v'(v_0, x)\!<\!0$, which means that any starting point from a certain neighbourhood of $v_0$ will make the ODE converge to $v_0$ itself.

As depicted in Fig. \ref{fig_imagpart}, $f(v, x)$ is not monotonic in $v$ (i.e., Problem \eqref{eq_problem} is non-convex), and within each lobe (e.g., the mainlobe or the sidelobe) of the antenna array pattern, there exists one stable point. The \emph{local optimal stable points} for the recursive procedure are given by
\begin{equation}\label{eq_stable_points}\begin{aligned}
\mathcal{S}(x)\!=\!&\left\{ v \in (-1, 1]: f(v, x) =0, f_v'(v, x) <0\right\} \\
       =\!&\left\{ v_k \in (-1, 1]: v_k = x + \frac{k\lambda}{(M-1)d}, k \in \mathbb{Z} \right\}.
\end{aligned}\end{equation}
Note that except for $x$, the antenna array gain is quite low at other local optimal stable points in $\mathcal{S}(x)$, where the loss of antenna array gain is nearly 20dB and will be higher if more antennas are configured. Hence, one key challenge is \emph{how to ensure that Algorithm 1 converges to the real direction $x$, instead of other local optimal stable points  in $\mathcal{S}(x)$.}
}

\begin{figure}[!t]
\centering
\includegraphics[width=7cm]{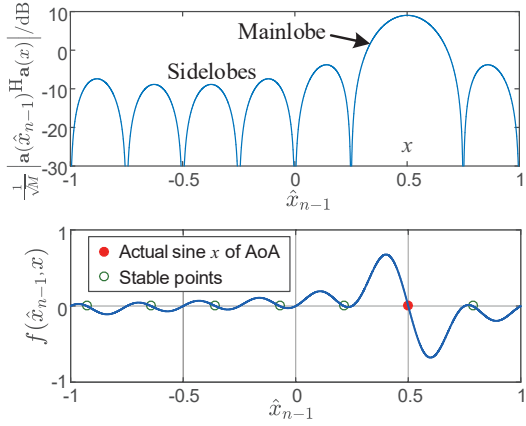}
\vspace{-2mm}
\caption{$\frac{1}{\sqrt{M}}\left|\textbf{a}(\hat{x}_{n-1})^\text{H}\textbf{a}(x)\right|$ and $f(\hat{x}_{n-1}, x)$ vs. $\hat{x}_{n-1}$ for $M = 8$, $x = 0.5$, $d = 0.5\lambda$. Note that the stable points are not at the sidelobe peaks.}
\vspace{-5mm}
\label{fig_imagpart}
\end{figure}

\subsection{Step-size Design and Asymptotic Optimality Analysis}\label{sec_performance_analysis}

In \emph{static} beam tracking scenarios, we adopt the widely used diminishing step-sizes, given by \cite{nevel1973stochastic,kushner2003stochastic, borkar2008stochastic}
\begin{equation}\label{eq_stepsize}
a_n = \frac{\alpha}{n + N_0}, ~~n = 1, 2, \ldots,
\end{equation}
where $\alpha > 0$ and $N_0 \ge 0$. We use the stochastic approximation and recursive estimation
theory \cite{nevel1973stochastic,kushner2003stochastic, borkar2008stochastic} to analyze Algorithm 1. In particular, we now develop a series of three theorems to resolve the challenge mentioned above.



\begin{theorem}[\textbf{Convergence to Stable Points}]\label{th_convergence}
\vspace{-1mm}
If $a_n$ is given by (\ref{eq_stepsize}) with any $\alpha > 0$ and $N_0 \ge 0$, then $\hat{x}_n$  converges to a unique point within $\mathcal{S}(x) \cup \{ -1\} \cup \{1 \}$ with probability one.
\vspace{-1mm}
\end{theorem}

\ifreport
\else

\begin{proof}[Proof Sketch]
This theorem is proven by applying Theorem 5.2.1 of \cite{kushner2003stochastic}. The proof is relegated to our technical report \cite{Li2017analog} due to space limitations.
\vspace{-1mm}
\end{proof}

\fi

Hence, for general step-size parameters $\alpha$ and $N_0$ in \eqref{eq_stepsize}, $\hat{x}_n$ converges to a stable point in $\mathcal{S}(x)$ or a boundary point. 

\begin{theorem}[\textbf{Convergence to the Real Direction $x$}]\label{th_lock}
\vspace{-1mm}
If (i) the initial point satisfies $\hat{x}_0 \in \mathcal{B}$, (ii) $a_n$ is given by (\ref{eq_stepsize}) with any $\alpha > 0$, 
then there exist  $N_0 \ge 0$ and $C(\hat{x}_0)>0$ such that
\begin{equation}\label{eq_lock}
P\left( \left. \hat{x}_n \rightarrow x \right| \hat{x}_0 \in \mathcal{B}(x) \right) \geq 1- 2e^{-C(\hat{x}_0)\frac{\rho}{\alpha^2}}.
\end{equation}


\vspace{-1mm}
\end{theorem}

\ifreport
\else

\begin{proof}[Proof Sketch]
Motivated by Chapter 4 of \cite{borkar2008stochastic}, we prove this theorem in three steps: in \emph{Step 1},  construct two continuous processes based on the discrete process $\left\{ \hat{x}_n \right\}$; in \emph{Step 2}, using these continuous processes, we form a sufficient condition for the convergence of the discrete process $\left\{ \hat{x}_n \right\}$;  in \emph{Step 3}, derive the probability lower bound for this condition, which is also a lower bound for $P\left( \left. \hat{x}_n\!\rightarrow\!x \right| \hat{x}_0\!\in\!\mathcal{B}\left(x\right) \right)$. 
The details are provided in \cite{Li2017analog}.
\vspace{-1mm}
\end{proof}

\fi

By Theorem \ref{th_lock}, if the initial point $\hat{x}_0$ is in the mainlobe $\mathcal{B}(x)$, the probability that $\hat{x}_n$ does not converge to $x$ decades \emph{exponentially} with respect to ${\rho}/{\alpha^2}$. Hence, one can increase the SNR $\rho$ and reduce the step-size parameter $\alpha$ to ensure $\hat{x}_n\!\rightarrow\!x$ with high probability. Under the condition of $\rho = 10\text{dB}$ and $M = 8\text{-}128$, typical values of $N_0$ required by the sufficient condition in Theorem \ref{th_lock} are 10-50. However, one can choose any $N_0\!\geq\!0$ to achieve a sufficiently high probability of $\hat{x}_n\!\rightarrow\!x$ in simulations.



\begin{theorem}[\textbf{Convergence to $x$ with the Minimum MSE}]
\vspace{-1mm}
\label{th_normal}
If (i) $a_n$ is given by (\ref{eq_stepsize}) with 
\begin{align}\label{eq_alpha}
\alpha = \frac{\lambda}{\sqrt{M}(M-1)\pi d} \overset{\Delta}{=} \alpha^*,
\end{align}
and any $N_0 \ge 0$, and (ii) $\hat{x}_n \rightarrow x$, 
then
\begin{equation}\label{eq_normal}
	\sqrt{n}\left(\hat{x}_n - x\right) \overset{d}{\rightarrow} \mathcal{N} \left(0, I_{\max}^{-1}\right),
\end{equation}
as $n \rightarrow \infty$, where $\overset{d}{\rightarrow}$ represents convergence in conditional distribution given $\hat{x}_n \rightarrow x$, and $I_{\max}$ is defined in \eqref{eq_opt_MMSE}. In addition, 
\begin{equation}\label{eq_normal2}
	\lim_{n\rightarrow\infty}~n~\mathbb{E}\left[\left(\hat{x}_n - x\right)^2\big| \hat{x}_n \rightarrow x\right] = I_{\max}^{-1}.
\end{equation}
\vspace{-1mm}
\end{theorem}

\ifreport
\else

\begin{proof}[Proof Sketch]
\vspace{-4.5mm}
This theorem is proven by applying Theorem 6.6.1 in \cite{nevel1973stochastic}. The details are provided in \cite{Li2017analog}.
\vspace{-1mm}
\end{proof}

\fi

Theorem \ref{th_normal} tells us that $\alpha$ should not be too small: If $\alpha=\alpha^*$ in \eqref{eq_alpha}, then the minimum CRLB on the RHS of \eqref{eq_CRLB} is achieved asymptotically with high probability, which ensures the highest convergence rate. In practice, we suggest to choose $\alpha=\alpha^*$ and $N_0=0$ in \eqref{eq_stepsize}.   
Interestingly, Theorem \ref{th_normal} can be readily generalized to the track of smooth functions of $x$:
\begin{corollary}\label{co_1}
\vspace{-1mm}
If the conditions of Theorem \ref{th_normal} are satisfied, then for any first-order differentiable vector function $\textbf{u}(x)$
\begin{align}
\!\lim_{n\rightarrow\infty} n~\!\mathbb{E}\left[\left\| \textbf{u}(\hat{x}_n)\!-\!\textbf{u}(x)\right\|^2_2\Big| \hat{x}_n \rightarrow x\right]\!=\!\left\|{ {\textbf{u}}'(x) }\right\|^2_2 I_{\max}^{-1}.
\end{align}
\vspace{-4.5mm}
\end{corollary}
\begin{proof}
See our technical report \cite{Li2017analog}.
\end{proof}

For example, consider the channel response $\textbf{{h}}(x) =  \beta \textbf{{a}}(x)$. If $\alpha = \alpha^*$ and $N_0=0$, 
Corollary \ref{co_1} tells us that, with high probability, the minimum CRLB of $\textbf{{h}}(x)$ is achieved in the following limit:
\begin{equation}\label{eq_CRLB_CR}
\begin{aligned}
	&~\lim_{n\rightarrow\infty}~n~\mathbb{E}\left[\left.\left\| \textbf{{h}}(\hat{x}_n) - \textbf{{h}}(x)\right\|^2_2 \right| \hat{x}_n \rightarrow x \right] \\
	= &~I_{\max}^{-1} \sum_{m=1}^{M-1}\left|\frac{\partial \left(\beta e^{-j \frac{2\pi d}{\lambda}m x}\right)}{\partial x}\right|^2 = \frac{(2M-1)\sigma^2}{3(M-1)}.
\end{aligned}
\end{equation}



\section{Simulation Results}\label{sec_simulation}
\ifreport
\else
\begin{figure}[!t]
\centering
\vspace{-3mm}
\includegraphics[width=6.3cm]{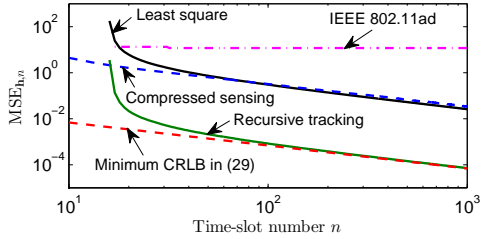}
\vspace{-2mm}
\caption{$\text{MSE}_{\textbf{h},n}$ vs. time-slot number $n$ in static beam tracking scenarios.}
\vspace{-5mm}
\label{fig_static_mse}
\end{figure}
\fi

We compare Algorithm 1 with three reference algorithms:
\begin{itemize}
\item[1)] \emph{IEEE 802.11ad \cite{IEEE80211ad}:} This algorithm contains two stages: beam sweeping and beam tracking. In stage one, sweep the beamforming directions in the DFT codebook (\ref{eq_codebook}) and choose the direction with the strongest received signal as the best beam direction. In stage two, probe the best beam direction and its two adjacent beam directions, then choose the strongest direction as the new best beam direction. The second stage is performed periodically.
\item[2)] \emph{Least square \cite{KaramiLS2007}:} Sweep all the beamforming directions in the DFT codebook  (\ref{eq_codebook}) and use the least square algorithm to estimate the channel response $\textbf{{h}}(x_n)$. Then obtain the analog beamforming vector $\textbf{w}_{n}$ for data transmission by
\vspace{-1mm}
\begin{equation}
	w_{mn} = \angle \hat{h}_m(x_n), m = 1, 2, \cdots, M,
	\vspace{-1mm}
\end{equation}
where $\hat{h}_m(x_n)$ is the $m$-th element of the estimated channel response $\hat{\textbf{{h}}}(x_n)$.

\item[3)] \emph{Compressed sensing \cite{Gao2015multi, Alkhateeb2015Compressed, Rial2016Hybrid}:}
Randomly choose the phase shifts $w_{mn}$ from $\{ \pm 1, \pm j \}$ to receive pilot signals. Then use the sparse recovery algorithm to estimate the sine $x_n$ of AoA, where a DFT dictionary with a size of 1024 will be used. 

\end{itemize}
Two performance metrics are considered: (i) the MSE of the channel response $\textbf{{h}}(x_n)$, defined by
\vspace{-1mm}
\begin{equation}
\begin{aligned}
\text{MSE}_{\textbf{h},n} \overset{\Delta}{=} \mathbb{E}\left[\left\| \hat{\textbf{{h}}}(x_n) - \textbf{{h}}(x)\right\|^2_2 \right] 
\end{aligned}
\vspace{-1mm}
\end{equation}
for the least square algorithm and
\vspace{-1mm}
\begin{equation}
\begin{aligned}
\text{MSE}_{\textbf{h},n} \overset{\Delta}{=} & ~\mathbb{E}\left[\left\| \textbf{{h}}(\hat{x}_n) - \textbf{{h}}(x)\right\|^2_2  \right]
\end{aligned}
\vspace{-1mm}
\end{equation}
for other algorithms, and (ii)
the achievable rate $R_n$, i.e.,
\vspace{-1mm}
\begin{equation}
	R_n \overset{\Delta}{=} \log_2 \left( 1 + \rho\left|\textbf{{w}}^\text{H}_n\textbf{{a}}(x_n)\right|^2 \right).
	\vspace{-1mm}
\end{equation}
The system parameters are configured as follows: $\beta\!=\!(1+j)/\sqrt{2}, {\rho}\!=\!|\beta|^2/\sigma^2\!=\!10\text{dB}, M\!=\!16, M_0\!=\!2M, d\!=\!0.5\lambda$. In the following subsections, we will investigate the static and dynamic beam tracking scenarios separately.

\subsection{Static Beam Tracking}\label{static_channel_simulation}

In static beam tracking scenarios, we assume that one pilot is allocated in each time-slot. Hence, these algorithms have the same pilot overhead. The received pilot signals of all time-slots $1,\ldots, n$ are used for estimating $x_n$ and $\textbf{{h}}(x_n)$ in the compressed sensing and least square algorithms. The step-size $a_n$ is given by (\ref{eq_stepsize}) with $\alpha = \alpha^*$ and $N_0=0$. The simulation results are averaged over 10000 random system realizations, where the beam direction $x$ is randomly generated by a uniform distribution on $[-1, 1]$ in each realization.

Figure \ref{fig_static_mse} plots the convergence performance of $\text{MSE}_{\textbf{h},n}$ over time. The MSE of Algorithm 1 converges quickly to the minimum CRLB given in \eqref{eq_CRLB_CR}, which agrees with Corollary \ref{co_1}, and is much smaller than those of IEEE 802.11ad, least square and compressed sensing algorithms.

\ifreport
\else
\begin{figure}[!t]
\centering
\vspace{-4mm}
\includegraphics[width=8.6cm]{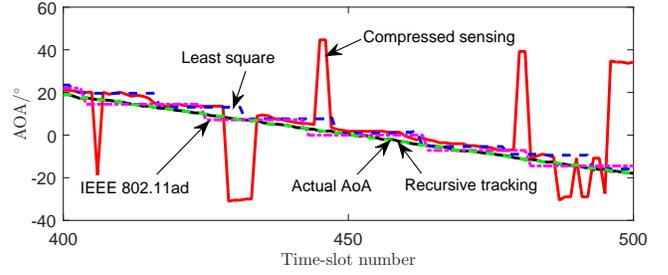}
\vspace{-2mm}\caption{AoA tracking in dynamic beam tracking scenarios.}
\vspace{-2mm}
\label{fig_dynamic_tracking}
\end{figure}

\begin{figure}[!t]
\centering
\includegraphics[width=8.6cm]{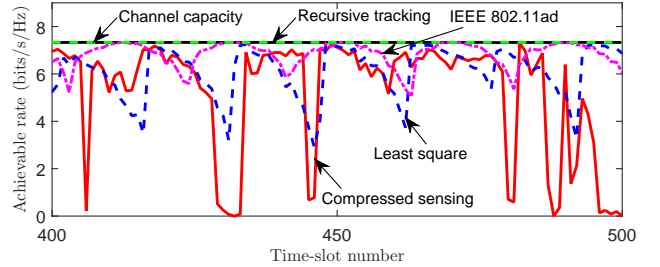}
\vspace{-2mm}
\caption{Achievable rate in dynamic beam tracking scenarios.}
\vspace{-5mm}
\label{fig_dynamic_tracking_rate}
\end{figure}
\fi

\subsection{Dynamic Beam Tracking}\label{dynamic_channel_simulation}\label{sec_dynamic_beam_tracking}

In dynamic beam tracking scenarios, where beam direction changes over time. In the beginning, we assume that continuous pilot training is performed and an initial estimate is obtained for all the algorithms. After that, one pilot is allocated in each time-slot to ensure that these algorithms have the same amount of pilot overhead.

The last $M/2$ pilot signals are used in the compressed sensing algorithm and the last $M$ pilot signals are used in the least square algorithm. For the IEEE 802.11ad algorithm, the probing period of its beam tracking stage is 3 time-slots. These parameters are chosen to improve the performance of these algorithms.
To keep track of the changing beam direction, the step-size $a_n$ of Algorithm 1 is fixed as
\vspace{-1mm}
\begin{equation}
a_n = \alpha^* =\frac{\lambda}{\sqrt{M}(M-1)\pi d},~\text{for all}~n \ge 1,
\vspace{-1mm}
\end{equation}
which is determined by the configuration of the antenna array and is independent of the SNR $\rho$.

%

Figures \ref{fig_dynamic_tracking} and \ref{fig_dynamic_tracking_rate} depict the AoA tracking and achievable rate performance in dynamic scenarios, where the AoA $\theta_n$ varies according to $\theta_n\!=\!({\pi}/{3})\sin\left({2\pi n}/{1000}\right)\!+\!0.005\vartheta_n$ with $\vartheta_n\!\sim\!\mathcal{N}(0, 1)$. Algorithm 1 always tracks the real AoA very well, and achieves the channel capacity $7.33$bits/s/Hz in all the time-slots. The performance of Algorithm 1 is much better than the other three algorithms, and the algorithm used by IEEE 802.11ad is better than the other two.

\ifreport
\else
Figures \ref{fig_dynamic_MSEvsV} and \ref{fig_dynamic_RatevsV} illustrate the average AoA tracking and achievable rate performance under a fixed angular velocity model $\theta_n=\theta_{n-1}+s_{n-1}\!\cdot\!\omega$ where $n=1,\ldots,10000$, $\theta_0=0$, $s_n\in\{-1,\!1\}$ denotes the rotation direction, and $\omega$ is a fixed angular velocity. The rotation direction $s_n$ is chosen such that $\theta_n$ varies within $[-{\pi}/{3},\!{\pi}/{3}]$. One can observe that Algorithm 1 can support higher angular velocities and data rates than the other algorithms when all 16 antennas are used. In addition, by using a subset of antennas, e.g., $M=4$ or 8, for beam tracking and all $16$ antennas for data transmission, the beam tracking regime of Algorithm 1 can be further enlarged. In \cite{Li2017analog}, it is shown that it is not always good to use fewer antennas for beam tracking. More specifically, when the SNR is $\rho=0$ dB, it is better to choose $M=8$ than $M=4$ because high antenna gain is needed at low SNR. 
\fi

\begin{figure}[!t]
\centering
\includegraphics[width=6cm]{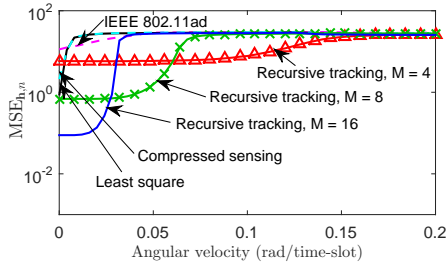}
\vspace{-2mm}
\caption{$\text{MSE}_{\textbf{h},n}$ vs. angular velocity in dynamic beam tracking scenarios.}
\vspace{-2mm}
\label{fig_dynamic_MSEvsV}
\end{figure}

\begin{figure}[!t]
\centering
\includegraphics[width=6cm]{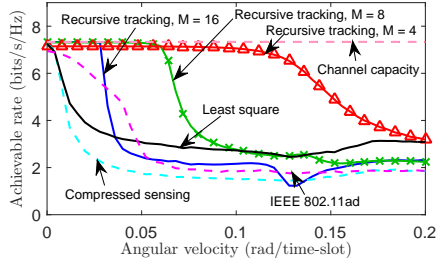}
\vspace{-2mm}
\caption{Data rate vs. angular velocity in dynamic beam tracking scenarios.}
\vspace{-5mm}
\label{fig_dynamic_RatevsV}
\end{figure}


According to Fig. \ref{fig_dynamic_RatevsV}, Algorithm 1 can achieve $95\%$ of the channel capacity when the angular velocity of the beam direction is 0.064rad/time-slot, the SNR is $\rho\!=\!10$dB, and $M\!=\!8$.  If each time-slot (TTI) lasts for 0.2ms (e.g., in 5G systems \cite{Pedersen20165g, Zong20165g}), Algorithm 1 can support an angular velocity of $0.064\!\times\!\frac{1000}{0.2}\!=\!320 \text{rad/s}\!\approx\!51 \text{circles/s}$. Similar calculation is made in \cite{Li2017analog} when $\rho\!=\!0$dB. Then, consider a TDMA pilot pattern where 1000 narrow beams are tracked by the antenna array periodically in a round-robin fashion such that 1 pilot is received in each time-slot. Algorithm 1 can support \textbf{18.33$^\circ$/s} (or \textbf{13.18$^\circ$/s}) per beam for tracking all these 1000 beams if $\rho\!=\!10$dB (or 0dB), which is \textbf{72mph} (or \textbf{52mph}) if the transmitters/reflectors steering these beams are at a distance of 100 meters. 

\ifreport

\else

\fi

\ifreport
\else
\vspace{-0.5mm}
\section{Conclusions}\label{conclusion}
\vspace{-0.5mm}
\fi
We have developed an analog beam tracking algorithm, and established its convergence and asymptomatic optimality. Our theoretical and simulation results show that the proposed algorithm can achieve much faster tracking speed, lower beam tracking error, and higher data rate than several state-of-the-art algorithms, with the same pilot overhead. In our future work, we will consider hybrid beamforming systems with multiple RF chains and two-dimensional antenna arrays, based on the methodology developed in the current paper.

\bibliographystyle{IEEEtran}
\bibliography{IEEEabrv,reference}

%
\IEEEpeerreviewmaketitle

\end{document}